\documentclass[a4paper]{article}
\usepackage[a4paper,hmargin=2.6cm,vmargin=2.8cm]{geometry}
\usepackage[utf8x]{inputenc}
\usepackage{amsmath,amsthm,amsfonts,amssymb}
\usepackage{feynmp}
\usepackage[bookmarksnumbered=true]{hyperref}
 \usepackage{dsfont}
\usepackage{graphicx}

\let\oldmarginpar\marginpar
\renewcommand\marginpar[1]{\-\oldmarginpar[\raggedleft\footnotesize #1]%
  {\raggedright\footnotesize #1}}

\newtheorem{thm}{Theorem}[section]  
\newtheorem*{thm*}{Theorem}
\newtheorem{cor}[thm]{Corollary}
\newtheorem{prp}[thm]{Proposition}
\newtheorem{lem}[thm]{Lemma}

\newtheorem*{rem}{Remark}
\newtheorem*{rems}{Remarks}

\renewcommand{\theequation}{\thesection.\arabic{equation}}

\newcommand{\tb}{{\tilde\beta}}
\newcommand{\order}{\mathcal{O}}

\newcommand{\di}{{\textrm{d}}}		
\newcommand{\lcal}{\ell}

\newcommand{\Rbb}{\mathbb{R}}		
\newcommand{\Cbb}{\mathbb{C}}		
\newcommand{\Nbb}{\mathbb{N}}		
\newcommand{\Zbb}{\mathbb{Z}}
\newcommand{\id}{\mathds{1}}
\newcommand{\tD}{\text{D}}
\newcommand{\lambdal}{{\Lambda_\lcal}}
\newcommand{\lambdaldual}{\Lambda_\lcal^{*}}
\newcommand{\tdl}{T^\tD_\lambdal}
\newcommand{\tr}{\operatorname{tr}}

\newcommand{\tagg}[1]{ \stepcounter{equation} \tag{\theequation} \label{eq:#1} } 

\newcommand{\be}{\begin{equation}}
\newcommand{\ee}{\end{equation}}

\def \bea{\begin{eqnarray}}
\def \eea{\end{eqnarray}}
\def\nn{{\nonumber}}

\newcommand{\gibbs}{\frac{e^{-\tb\tdl}}{\tr e^{-\tb \tdl}}}

\newcommand{\gs}{\lvert\text{gs}\rangle}

\title{Interaction Corrections to Spin-Wave Theory in the Large-S Limit of the Quantum Heisenberg Ferromagnet}

\author{Niels Benedikter\vspace{1em}\\QMATH\\Department of Mathematical Sciences,  University of Copenhagen\\Universitetsparken 5, 2100 Copenhagen, Denmark\vspace{.25em}}

\begin{document}
\begin{fmffile}{diags}
\maketitle

\begin{abstract}
The Quantum Heisenberg Ferromagnet can be naturally reformulated in terms of interacting bosons (called spin waves or magnons) as an expansion in the inverse spin size. We calculate the first order interaction correction to the free energy, as an upper bound in the limit where the spin size $S \to \infty$ and $\beta S$ is fixed ($\beta$ being the inverse temperature). Our result is valid in two and three spatial  dimensions.

We extrapolate our result to compare with Dyson's low-temperature expansion. While our first-order correction has the expected temperature dependance, in higher orders of the perturbation theory cancellations are necessary.
\end{abstract}

\section{Introduction and Main Result}
We consider the ferromagnetic quantum Heisenberg model on a volume $\Lambda_L = [0,L]^d \cap \Zbb^d$ in spatial dimension $d=2$ or $3$, with nearest neighbour interaction. It is described by the Hamilton operator
\[H_{\Lambda_L} := \sum_{\langle x, y\rangle \subset \Lambda_L}\left(S^2 - {{S}}_{x} \cdot {{S}}_{y} \right)\]
acting on the Hilbert space $\bigotimes_{x\in\Lambda_L} \Cbb^{2S+1}$.
The sum is over unordered nearest neighbour pairs in $\Lambda_L$, and ${{S}}_x = (S^1_x,S^2_x,S^3_x)$ is a spin-$S$ operator (where $2S \in \Nbb$), i.\,e.\ satisfying the commutation relations
$[{S}^j_{x},{S}^k_{y}] = \delta_{x,y}\,i \sum_{l=1}^3 \varepsilon_{jkl} {S}^l_{x}$ for $j,k \in \{1,2,3\},\ x,y \in \Lambda_L$,
(with $\varepsilon_{jkl}$ the totally antisymmetric symbol) and the condition $({{S}}_{x})^2 = ( S^1_{x})^2 + ( S^2_{x})^2 + ( S^3_{x})^2 = S(S+1)$. 
The quantity we are interested in is the free energy in the thermodynamic limit,
\[f(S,\beta) := \lim_{\Lambda_L \to \infty} f(S,\beta,\Lambda_L), \quad f(S,\beta,\Lambda_L) := -\frac{1}{\beta L^d} \log \tr e^{-\beta H_{\Lambda_L}},\]
where $\beta$ is the inverse temperature.
This model is of great importance for the understanding of ferromagnetism, which poses a challenging mathematical problem, namely the spontaneous breaking of continuous symmetries.

\medskip

In this note we are interested in the spin-wave approximation, which goes back to Bloch in 1930 \cite{Bloch1930}. Bloch noticed that the low-energy excitations of the Heisenberg model can be approximately described as independent bosonic modes with energy given by the dispersion relation $\varepsilon(k) := \sum_{j=1}^d 2(1-\cos(k_j))$, and thus was able to predict the behaviour of the thermodynamic quantities at low temperature.

Let us recall this theory in modern language. The ground states of the Heisenberg model at zero temperature are states in which all spins point parallel, while the direction is arbitrary. W.\,l.\,o.\,g.\ we can take the ground state to have all spins pointing in the $-z$-direction, i.\,e.\ 
\[\gs = \bigotimes_{x\in\Lambda_L} \lvert -S\rangle_x,\quad\text{where } \lvert -S\rangle_x \in \Cbb^{2S+1},\quad S^3_x \lvert -S\rangle_x = -S \lvert -S\rangle_x.\]
Spin waves (or magnons) are excitations of the form
\[\lvert k\rangle := (2S L^d)^{-1/2}\sum_{x\in\Lambda_L} e^{ik\cdot x} S^+_x \gs =: (2S)^{-1/2} S^+_k \gs,\]
with a momentum $k \in \frac{2\pi}{L}\Zbb^d$ (taking in this introduction for simplicity periodic boundary conditions). These states are normalized eigenstates of the Hamiltonian, 
\[H_{\Lambda_L}\lvert k \rangle = S \varepsilon(k) \lvert k\rangle.\]
Treating the system as a system of non-interacting bosons with energy $S\varepsilon(k)$, one arrives at the following prediction for the free energy:
\begin{equation}\label{eq:freeenergy}f(S,\beta) \simeq \frac{1}{\beta} \int_{[-\pi,\pi]^d} \frac{\di^d k}{(2\pi)^d} \log(1-e^{-\beta S \varepsilon(k)}) \sim \beta^{-\frac{d+2}{2}} S^{-\frac{d}{2}} \int \frac{\di^d k}{(2\pi)^d} \log( 1 - e^{-k^2}) \quad \text{as } \beta \to \infty.\end{equation}
Unfortunately, the spin-wave excitations do not actually behave like independent bosonic modes; the states with more than one spin wave---i.\,e.\ constructed by applying more than one operators $S^+_k$---are neither eigenstates of the Hamiltonian nor orthogonal.
 This problem is treated more systematically by the Holstein-Primakoff mapping \cite{HP40} of spin operators on bosonic operators: defining the spin raising and lowering operators $ S^+_x :=  S^1_x + i  S^2_x$ and $ S^-_x :=  S^1_x - i  S^2_x$, one introduces the mapping on bosonic creation and annihilation operators $a^*_x$ and $a_x$ by
\begin{equation}\label{eq:hpmapping} S^+_x = \sqrt{2S}a^*_x \left[ 1-\frac{a^*_x a_x}{2S} \right]^{1/2}, \quad S^-_x = \sqrt{2S}\left[ 1-\frac{a^*_x a_x}{2S} \right]^{1/2} a_x, \quad S^3_x = a^*_x a_x - S.\end{equation}
The creation and annihilation operators act on the subspace of bosonic Fock space where the number of bosons per lattice site $x \in \Lambda_L$, $n_x := a^*_x a_x$, is restricted to be at most $2S$. The Hamiltonian becomes
\begin{equation}\label{eq:hphamiltonian}H_{\Lambda_L} = S \sum_{\langle x,y\rangle \subset \Lambda_L} \left( -a^*_x\sqrt{1-\frac{n_x}{2S}}\sqrt{1-\frac{n_y}{2S}} a_y - a^*_y \sqrt{1-\frac{n_y}{2S}}\sqrt{1-\frac{n_x}{2S}} a_x + n_x + n_y - \frac{1}{S} n_x n_y\right).\end{equation}
The Hamiltonian is then formally expanded in $\frac{n_x}{2S}$, a procedure that is expected to lead to good approximations if $S$ is large or if the expected occupation numbers are small, i.\,e.\ at low temperature. The leading term is given by
\[H_{\Lambda_L} \simeq S \sum_{\langle x, y\rangle\subset\Lambda_L} (a^*_x - a^*_y)(a_x - a_y),\]
the second quantized Laplacian on the lattice, giving rise to the free-boson picture and the expressions \eqref{eq:freeenergy}. In this paper, we are interested in the residual interaction between spin waves as given through the higher orders terms of the expansion of the Hamiltonian. The corrections due to interactions have given rise to considerable discussions in the physics community, a topic that we shall discuss further after the statement of our theorem.

\medskip

There are three scaling regimes which are important in the study of the Heisenberg model. Physically most important is the limit of low temperature $\beta \to \infty$ and fixed $S$; it is however very difficult to study. In fact, only recently has the leading order of the free energy been rigorously derived \cite{CGS14} (non-optimal bounds were proved earlier in \cite{CS91,T93}). On the other extreme there is the classical scaling regime, defined by $\beta S^2$ fixed and $S\to \infty$. In this scaling, convergence to the classical Heisenberg model has been proven \cite{Lieb73} (for non-zero magnetic field also in \cite{CS90}). For the classical Heisenberg model, it was proven that it has a critical temperature of order unity \cite{FSS76}. This suggests that by going to the intermediate scaling regime
\begin{equation}\label{eq:intermediateregime}\tb := \beta S \text{ fixed}, \quad S \to \infty\end{equation}
we can study the ferromagnetic phase in a regime that is more accessible (since the attractive interaction between spin waves in this regime is of order $S^{-1}$) than the low temperature regime, and still governed by quantum theory. For $d=3$, the leading order of the free energy in this regime has been obtained in \cite{CG12}. (This regime was introduced and compared to the classical regime in \cite{CS90}, where the leading order in the case of non-zero magnetic field was derived. In this context also the important random walk representation of the Heisenberg model was developed \cite{CS91rw, CS91class}.)

\medskip

In this note we calculate an upper bound on the free energy which includes interaction effects to first order in the intermediate scaling regime \eqref{eq:intermediateregime}. Our result is valid in two and three dimensions.
\begin{thm*}\label{thm:mainresult}
Let $d = 2$ or $3$, and $\tb$ fixed and sufficiently large. Then the free energy is bounded above by
 \[\frac{f(S,\beta)}{S} \leq \frac{1}{\tilde \beta} \int_{[-\pi,\pi]^d} \frac{\di^d k}{(2\pi)^d} \log \left( 1-e^{-\tilde \beta \varepsilon(k)} \right) - \frac{1}{S} \frac{1}{4d} \left( \int_{[-\pi,\pi]^d} \frac{\di^d k}{(2\pi)^d} \frac{\varepsilon(k)}{e^{\tilde \beta \varepsilon(k)}-1}\right)^2 + \frac{r_d(S,\tb)}{S^2}. \]
 for $S \to \infty$, where $\varepsilon(k) = \sum_{j=1}^d 2(1-\cos(k_j))$ is the spin-wave dispersion relation.
 
 The error term\footnote{We use $C$ for constants independent of $S$ and $\tb$ (and $\ell$, to be introduced later), with value possibly changing from line to line.} is $r_2(S,\tb) \leq C \tb^{-2} (\log S\tb)^3$ and $r_3(S,\tb) \leq C \tb^{-3}$ (independent of $S$).
 \end{thm*}
%
The first summand---the leading order---describes free bosons (the spin waves) on the lattice. The second summand is the first order correction due to the interaction of spin waves. The interaction corrections and their  temperature dependence have long been controversial among physicists, with many contradictory corrections proposed, e.\,g.\ of order $\beta^{-3}$ or $\beta^{-11/4}$ \cite{Kramers36,Opechowski37,Schafroth54,Kranendonk54}. Eventually Dyson mostly settled the issue in his landmark papers \cite{Dyson1,Dyson2}, arguing (for $d=3$) that the correction is very small 
 for low temperature\footnote{Here we are rewriting the equation  \cite[(131)]{Dyson2}. The constant summands $-\frac{1}{2}JS^2\gamma_0$ and $-LS$ given in Dyson's paper vanish in our setting, since we shifted the ground state energy by the $S^2$ in the Hamiltonian, and since we do not have an external magnetic field. The next three summands in Dyson's paper are an expansion of the leading integral written here. Dyson's $- k T C_3 S^{-1}[Z_{5/2}(\beta L)]^2 \theta^4$ corresponds to the term written here in square brackets; the $\mathcal{O}(S^{-2})$ is due to the $S$-dependence of $C_3$.}, namely of order $\tilde \beta^{-5}$: 
\begin{equation}\label{eq:star}\begin{split}\frac{f(S,\beta)}{S} & = \frac{1}{\tilde \beta} \int_{[-\pi,\pi]^3} \frac{\di^3 k}{(2\pi)^3} \log \left( 1-e^{-\tilde \beta \varepsilon(k)} \right) - \left[ S^{-1}\frac{3}{128(2\pi)^3}\zeta(5/2)^2 + \mathcal{O}(S^{-2})\right] \tb^{-5} + \mathcal{O}(\tilde\beta^{-11/2}).\nn\end{split}\end{equation}
Nevertheless, the temperature dependence of the interaction corrections is being studied up to recent years, mostly substantiating Dyson's result by other formal methods, e.\,g.\ an effective Lagrangian method \cite{Hofmann02, Hofmann11}. 
One paper that should be highlighted is \cite{Zittartz65}, which lead to Dyson's result by a less cumbersome method of introducing additional bosonic degrees of freedom coupled to the spin system. 
However, there is also work newer than Dyson's papers which contradicts Dyson's result; see \cite{Hewson63} and the list of references therein.

To compare our result with Dyson's result we now formally think of $\tilde \beta \to \infty$ in our result. Expanding $\varepsilon(k)$ for small $k$ we get
\[-\frac{1}{S}\frac{1}{12} \left( \int_{[-\pi,\pi]^3} \frac{\di^3 k}{(2\pi)^3} \frac{\varepsilon(k)}{e^{\tilde \beta \varepsilon(k)}-1} \right)^2 \simeq -S^{-1}\frac{1}{12}  \frac{1}{(2\pi)^6}\left( \int_{\Rbb^3} \di^3 k' \frac{\lvert k'\rvert^2}{e^{\lvert k'\rvert^2}-1} \right)^2 \tb^{-5}.\]
Using spherical coordinates and $(e^{\lvert k'\rvert^2}-1)^{-1} = \sum_{n=1}^\infty e^{-n\lvert k'\rvert^2}$, we recover Dyson's result at order $S^{-1}$.

Of course this argument is beyond the proven validity of our theorem because $r_3(S,\tb) \sim \tb^{-3}$. Inspection of our proof shows that $r_3$ actually consists of two kinds of errors: corresponding to Dyson's kinematical interaction we have errors controlled by $\tb^{-S}$, and corresponding to Dyson's dynamical interaction we have our main error of order $S^{-2}\tb^{-3}$ (see Lemma $\ref{lem:II}$). However, there is a cancellation mechanism which is supposed (but not proven) to make the latter as small as $\tb^{-5}$ and which we discuss perturbatively in the Appendix.

\begin{rems}
\begin{enumerate}\item Our method also provides a partial result for the case of $\tb$ small, but not too small. As the prime example we choose $\tb = S^{-\alpha}$, $\alpha\in [0,1)$, and  with minor changes (see Remark (ii) after Lemma \ref{lem:oneminusp}) obtain
$r_3(S,\tb)/S^2 = \mathcal{O}(S^{3\alpha-2})$ for $S\to\infty$ (for $d=2$ with logarithmic correction). For comparison: in this case the leading term of $f(S,\beta)/S$ is $\tb^{-1}(c_0 \log \tb + c_1 + c_2\tb + \ldots ) \simeq S^\alpha \log S$ and the first order correction is $S^{-1}\tb^{-2} \simeq S^{2\alpha-1}$.
\item While for $d=3$ the validity of spin-wave theory has long been trusted in by physicists, it remained more disputed in $d=2$. Our result supports the validity of spin-wave theory in $d=2$, as far as the free energy in the intermediate scaling regime is concerned. Notice that also the (leading order) lower bound from \cite{CG12} is easily checked to be valid also for two dimensions.
\item Obtaining the first order correction as a \emph{lower} bound remains open; even in the intermediate scaling regime this is expected to be a very difficult problem.
\end{enumerate}
\end{rems}


\section{Proof}
Our proof adapts the methods used recently in \cite{CG12, CGS14}. We use the Gibbs variational principle and the bosonic representation of the Heisenberg model in terms of spin-waves due to Holstein and Primakoff \cite{HP40}. Our trial state is a bosonic quasifree state that we have to supplement with a cutoff on the number of bosons. In our proof we will first remove the cutoff of the particle number. Thereafter we can use Wick's theorem to calculate expectation values, which enables us to bound the error terms and calculate the correction.

\medskip

 To get an upper bound, we use the Gibbs variational principle, which states that
\[f(S,\beta,\Lambda) \leq \frac{1}{\lvert \Lambda \rvert} \tr H_\Lambda \Gamma + \frac{1}{\beta\lvert\Lambda\rvert} \tr \Gamma \log \Gamma\]
for all positive trace class operators $\Gamma$ normalized to $\tr \Gamma = 1$ (i.\,e.\ states).

\medskip

Following a standard procedure, we first use the Gibbs variational principle to break up the system into smaller boxes with Dirichlet boundary conditions. We assume that $L = k(\lcal+1)$ for some integers $k$ and $\lcal$. 
On the set $\mathcal{C} := \{x \in \Lambda_L: x_i = n(\ell+1) \text{ for some }i=1,\ldots d \text{ and } n \in \Zbb\}$, we restrict the spins in our states to $S^3_x= -S$.
Now $\Lambda_L \setminus \mathcal{C}$ is a union of translates of boxes $\Lambda_\ell := [1,\ell]^d \cap \Zbb^d$, on which the Hamiltonian on the restricted class of states becomes 
\begin{equation}H^\text{D}_{\Lambda_\ell} = H_{\Lambda_\ell} + \sum_{x \in \partial\Lambda_\ell} \left( S^2 + S S_x^3 \right).\label{eq:Dirichlet}\end{equation}
Here the boundary $\partial\Lambda_\ell$ consists of the points in $\Lambda_\ell$ having distance $1$ from $\mathcal{C}$. The extra summand is non-negative, and so $H^\text{D}_{\Lambda_\ell} \geq H_{\Lambda_\ell}$.
Due to this extra Dirichlet restriction on the states, the variational principle yields
\[f(S,\beta,\Lambda_L) \leq (1+\ell^{-1})^{-d} f^\text{D}(S,\beta,\Lambda_\ell), \quad \text{with } f^\text{D}(S,\beta,\Lambda_\lcal) := -\frac{1}{\beta \ell^d} \log \tr e^{-\beta H^\text{D}_{\Lambda_\lcal}}.\]
Letting $k\to\infty$, we obtain the following bound for the thermodynamic limit:
\[f(S,\beta) \leq (1+\ell^{-1})^{-d} f^\text{D}(S,\beta,\Lambda_\ell).\]
This bound holds for any integer $\lcal$, and we will later choose $\lcal = \tb^d S^2$ (or more precisely the nearest integer) to optimize the error bounds.

\medskip

The next step is to rewrite the Hamiltonian \eqref{eq:Dirichlet} through the Holstein-Primakoff mapping \eqref{eq:hpmapping}. Leaving aside for a moment the Dirichlet boundary condition, recall the Hamiltonian \eqref{eq:hphamiltonian}. 
We consider the formal Taylor expansion w.\,r.\,t.\ the small parameter $1/S$,
\begin{equation}\label{eq:formalexpansion}\begin{split}
   H_{\Lambda_\ell} & = S \bigg[\sum_{\langle x,y\rangle \subset \Lambda_\ell} \left( -a^*_x a_y - a^*_y a_x + n_x + n_y \right) \\
	     &\quad\quad\  + \frac{1}{S}\sum_{\langle x,y\rangle \subset \Lambda_\ell} \left( a^*_x \left(\frac{n_x}{4}+\frac{n_y}{4}\right) a_y + a^*_y \left(\frac{n_y}{4}+\frac{n_x}{4} \right) a_x - n_x n_y \right) + R_{\Lambda_\ell} \bigg]\\
	     & =: S \left( T_{\Lambda_\ell} + I_{\Lambda_\ell} + R_{\Lambda_\ell} \right) =: H_{0,\Lambda_\ell} + S R_{\Lambda_\ell}.
  \end{split}
\end{equation}
Here $T_{\Lambda_\ell}$ contains the terms which are formally of order unity and $I_{\Lambda_\ell}$ the terms formally of order $S^{-1}$, and $R_{\Lambda_\ell}$ is the remainder. The term $T_{\Lambda_\ell}$ is the second quantization of the discrete Laplacian. We will show that $I_{\Lambda_\ell}$ gives the interaction correction, whereas the contribution of $R_{\Lambda_\ell}$ is estimated to be of order $1/S^2$ and thus negligible.

Including the Dirichlet boundary term from \eqref{eq:Dirichlet} in the Laplace operator, i.\,e.\ setting $T^\text{D}_{\Lambda_\ell} := T_{\Lambda_\ell} + \sum_{x \in \partial\Lambda_\ell} \big( S + S_x^3 \big)$, we similarly get  $H^\text{D}_{\Lambda_\ell} = S \left( T^\text{D}_{\Lambda_\ell} + I_{\Lambda_\ell} + R_{\Lambda_\ell}\right) =: H^\text{D}_{0,\Lambda_\ell} + S R_{\Lambda_\ell}$. In particular, $I_{\Lambda_\ell}$ and $R_{\Lambda_\ell}$ are unchanged by the addition of Dirichlet boundary conditions to the Hamiltonian.

\medskip

In our proof we will obtain an upper bound on $f^\text{D}(S,\beta,\lambdal)$ using the trial state
\begin{equation}\label{eq:trialstate}
\Gamma^\text{D} = \frac{P e^{-\tb T^\text{D}_{\Lambda_\lcal}} P}{ \tr e^{-\tb T^\text{D}_{\Lambda_\lcal}} P}, \quad P = \prod_{x \in \Lambda_\lcal} \id(n_x \leq 2S), 
\end{equation}
in the Gibbs variational principle for $f^\text{D}(S,\beta,\lambdal)$. The projection $P$ ensures that we are in the subspace of bosonic Fock space where there are at most $2S$ particles per site; thus it is valid to use the bosonic formulas for the Hamiltonian equivalently to the formulas in terms of spin operators. (In  \cite{CGS14} a similar trial state was used, projecting on occupation numbers $n_x \leq 1$; projecting on $n_x \leq 2S$ has the advantage of giving exponential decay w.\,r.\,t.\ $S$ in Lemma \ref{lem:oneminusp}.)

\subsection{Bounding the Error Terms}\label{sec:errorterms}
By Gibbs' variational principle
\[f^\text{D}(S,\beta,\Lambda_\lcal) \leq \frac{1}{\lcal^d} \tr H^\text{D}_{\Lambda_\lcal} \Gamma^\text{D} + \frac{1}{\beta \lcal^d} \tr \Gamma^\text{D} \log \Gamma^\text{D}.\]
Inserting the expression \eqref{eq:trialstate} for the trial state we obtain
\begin{align*}f^\text{D}(S,\beta,\Lambda_\lcal) & \leq \frac{S}{\lcal^d} \tr \left(T^\tD_\lambdal + I_\lambdal + R_\lambdal\right) \frac{P e^{-\tb T^\tD_\lambdal}P}{\tr e^{-\tb \tdl} P} + \frac{1}{\beta \lcal^d} \tr \frac{P e^{-\tb \tdl} P}{\tr e^{-\tb \tdl} P} \log P e^{-\tb \tdl} P \\ & \quad - \frac{1}{\beta \lcal^d} \log \tr e^{-\tb \tdl} P.
\end{align*}
From \cite[(4.22)]{CGS14} we have the inequality
\[\begin{split}
   \tr P e^{-\tb \tdl} P \log P e^{-\tb \tdl} P & = \tr e^{-\tb \tdl/2} P e^{-\tb \tdl/2} \log e^{-\tb \tdl/2} P e^{-\tb \tdl/2} \\
   & \leq \tr e^{-\tb \tdl/2} P e^{-\tb \tdl/2} \log e^{-\tb \tdl} = -\tb \tr P e^{-\tb \tdl}\tdl.
  \end{split}
\]
For the expectation values and their normalization we introduce the notation
\[\langle \cdot\rangle = \tr \cdot\,\frac{e^{-\tb \tdl}}{\tr e^{-\tb \tdl}}, \quad \langle \cdot\rangle_P = \tr \cdot\, \frac{P e^{-\tb \tdl} P}{\tr P e^{-\tb \tdl}}, \quad N_P = \frac{\tr e^{-\tb \tdl}}{\tr P e^{-\tb \tdl}}.\]
Dividing by $S$, we thus obtain
\begin{align*}\frac{f^\text{D}(S,\beta,\Lambda_\lcal)}{S} & \leq \frac{1}{\lcal^d} \langle T^\tD_\lambdal + I_\lambdal \rangle_P + \frac{1}{\lcal^d} \langle R_\lambdal \rangle_P  - \frac{1}{\lcal^d} \frac{\tr \tdl e^{-\tb \tdl} P}{\tr e^{-\tb \tdl} P} - \frac{1}{\tb \lcal^d} \log \tr e^{-\tb \tdl} P\\
& =: \text{I} + \text{II} + \text{III} + \text{IV}.
\end{align*}

Before analysing terms I through IV, we establish some crucial lemmas. We reprove bounds on the expected number of bosons at site $x \in \lambdal$, $\rho(x) = \langle n_x\rangle = \tr n_x e^{-\tb \tdl}/\tr e^{-\tb \tdl}$ (c.\,f.\ \cite{CGS14} for another proof for $d=3$), and use them to show that $\langle 1-P\rangle$ is exponentially decaying as $S\to \infty$.

\begin{lem}\label{lem:rhobound}
The number of bosons at lattice site $x$, $\rho(x) = \langle n_x\rangle$, is bounded  by
\begin{equation}\label{eq:numberbound}\sup_{x \in \lambdal} \rho(x) \leq \frac{\pi^{3/2}}{8} \zeta(3/2) \tb^{-3/2} \ (d=3) \quad \text{and}\quad \sup_{x \in \lambdal}\rho(x) \leq 4\pi \tb^{-1} \log(\ell) \ (d=2).\end{equation}
(For $d=2$ we have to assume $2\tb >1 > 2\tb/(\ell+1)$. The constant $4\pi$ is a rather rough estimate.)
\end{lem}
\begin{proof}
We use the Fourier transform of the creation and annihilation operators. It is given by
\begin{equation}
\label{eq:deffouriertransform}a^*_x = \sum_{k \in \lambdaldual} \varphi_k(x) a^*_k, \quad a_x = \sum_{k \in \lambdaldual} \varphi_k(x) a_k,\quad \lambdaldual =  \frac{\pi}{\lcal+1}\{1,2,\ldots,l\}^d,\end{equation}
where the $\varphi_k$ are the orthonormal eigenfunctions of the discrete Laplacian with Dirichlet boundary conditions on the box $\lambdal$, i.\,e.\ 
$\varphi_k(x) = 2^{d/2} (\ell+1)^{-d/2} \prod_{j=1}^d \sin\left(x_j k_j\right)$. Their normalization is such that
$\frac{2}{\ell+1}\sum_{x_i=1}^{\ell} \sin\left(x_i k_i\right) \sin\left(x_i k'_i\right) = \delta_{k_i,k'_i}$ and $\sum_{x\in \lambdal} \varphi_k(x) \varphi_{k'}(x) = \delta_{k,k'}$. Thus we find
\begin{equation}\label{eq:rhoboundcalc}\rho(x) = \tr n_x \frac{e^{-\tb \tdl}}{\tr e^{-\tb\tdl}} = \sum_{k \in \lambdaldual} \lvert \varphi_k(x)\rvert^2 \frac{1}{e^{\tilde \beta \varepsilon(k)}-1}\leq \left(\frac{2}{\pi}\right)^d \sum_{k \in\lambdaldual} \left(\frac{\pi}{\ell+1}\right)^d \frac{1}{e^{\tilde\beta \varepsilon(k)}-1}.\end{equation}

Case $d = 3$: The function $f(k):=(e^{\tb \varepsilon(k)}-1)^{-1}$ is monotonically decreasing in $k_1$, $k_2$ and $k_3$; thus  $(\frac{\pi}{\ell+1})^3\sum_{k \in\lambdaldual} f(k)$ is a lower Riemann sum and as such bounded above by the integral of $f$ on $[0,\pi]^3$. The integral is seen to be order $\tb^{-3/2}$ using  $\varepsilon(k) \geq 4 \lvert k\rvert^2/\pi^2$ and substituting $k' = \tb^{1/2} \frac{2}{\pi}k$; the numerical constant is obtained switching to spherical coordinates and using $\int_0^\infty k^2(\exp(k^2)-1)^{-1}\di k = \frac{\sqrt{\pi}}{4} \zeta(3/2)$.

Case $d=2$: Again the sum is a lower Riemann sum. However, the integral diverges at the origin, so we only use it as an upper bound outside the box $[0,\frac{\pi}{\ell+1}]^2$; inside the box we keep the original summand:
\[\rho(x) \leq \left(\frac{2}{\pi}\right)^2 \int_{[0,\pi]^2 \backslash B_{\pi/(\ell+1)}(0)}\di^2 k \frac{1}{e^{\tb \varepsilon(k)}-1} + \left(\frac{2}{\pi}\right)^2 \frac{1}{e^{\tb \varepsilon\left((\frac{\pi}{\ell+1},\frac{\pi}{\ell+1})\right)}-1} \left(\frac{\pi}{\ell+1} \right)^2.\]
(We have actually enlarged the integral a bit by not excluding $[0,\pi/(\ell+1)]^2$ but only the ball $B_{\pi/(\ell+1)}(0)$ to simplify the further estimates.) Some basic rough estimates yield the constant.
%
\end{proof}

\begin{rem}
 For $d=3$, the estimate can be slightly improved to $\sup_{x \in \Lambda_\ell} \rho(x) \leq (2/\pi)^{3/2} \zeta(3/2) \tb^{-3/2}$, estimating the lattice heat kernel through the asymptotics of a modified Bessel function \cite[(4.15)ff]{CGS14}. However, this approach fails for $d=2$, and for this reason we stick with the estimate presented here.
\end{rem}

\begin{lem}[Exponential Decay of $\langle 1-P\rangle$]\label{lem:oneminusp}
With $C$ the respective constant from \eqref{eq:numberbound}, we have
\[\langle 1-P\rangle \leq e \ell^3 (2S+1) \left(C\tb^{-3/2}\right)^{2S} \ (d=3)\quad \text{and}\quad \langle 1-P\rangle \leq e \ell^2 (2S+1) \left( C \tb^{-1} \log (\ell) \right)^{2S}\ (d=2),\]
i.\,e.\ for large enough $\tb$ we have exponential decay as $S \to \infty$.
\end{lem}
\begin{rems} \begin{enumerate}\item With our later choice $\ell = S^2 \tb^{d}$, this Lemma provides exponential decay of $\langle 1-P\rangle$ as $S \to \infty$. Our method fails for $d=1$ since then $\rho(x)\sim \ell$ and consequently we lose the exponential decay.
\item For small $\tb$ and $d=3$, we find the better bound $\rho(x) \leq 8\pi \tb^{-1}$ by expanding the exponential in \eqref{eq:rhoboundcalc}.
For the particular case $\tb = S^{-\alpha}$, this bound implies $\langle 1-P\rangle \leq C \ell^3 (2S+1) e^{-S^{1-\alpha}/4\pi}$. In this case the best pick for $\ell$ turns out to be $\ell = S^2$.
\end{enumerate}
\end{rems}
\begin{proof}[Proof of Lemma \ref{lem:oneminusp}]
Recall that $P = \prod_{x\in\lambdal} \id(n_x \leq 2S)$. Thus its expectation value is the probability that on all lattice sites $x$ we have $n_x \leq 2S$, i.\,e.\ $\langle P \rangle = \mathbb{P}(\forall x\in\lambdal :\ n_x\leq 2S)$. Consequently
\begin{align*}\langle 1-P \rangle & = \mathbb{P}(\exists x\in\lambdal :\ n_x > 2S) \leq \sum_{x\in\lambdal} \mathbb{P}(n_x > 2S) = \sum_{x\in\lambdal} \langle \id(n_x>2S) \rangle. \tagg{secondline}
  \end{align*}
We bound the step function by an exponential to see that for all $\lambda >0$ we have $\langle \id(n_x>2S) \rangle \leq \langle e^{\lambda n_x} \rangle e^{-\lambda 2S}$. 
Now denoting by $:\!\cdot\!:$ the normal-ordered product (i.\,e.\ all $a^*$s to the left of the $a$s), we have $\langle e^{\lambda n_x}\rangle = \langle :\!\exp(g(\lambda) a^*_x a_x)\!: \rangle$; this is proven taking the trace in a basis of eigenvectors of $n_x$. Expanding the normal-ordered exponential and using that by Wick's theorem $\langle a^*_x \cdots a^*_x a_x \cdots a_x\rangle = n! \langle a^*_x a_x \rangle^n$, we obtain
\[\langle e^{\lambda n_x} \rangle = \frac{1}{1-g(\lambda)\rho(x)},\]
where $g(\lambda) := e^\lambda-1$. Minimizing $\left(1-g(\lambda)\rho(x)\right)^{-1}e^{-\lambda 2S}$ w.\,r.\,t.\ $\lambda$, we easily find\footnote{For points $x \in \lambdal$ with $\rho(x) = 0$ we can't minimize, but the estimate is trivially true.}
\[\langle 1-P\rangle \leq \sum_{x\in\lambdal} \frac{2S+1}{1+\rho(x)} \left(\frac{2S+1}{2S} \frac{\rho(x)}{1+\rho(x)}\right)^{2S} \leq \sum_{x\in\lambdal} (2S+1)e \rho(x)^{2S}.\]
Finally we use Lemma \ref{lem:rhobound} to bound $\rho(x)$.
\end{proof}
As a corollary we prove that $N_P \simeq 1$, up to an error of order $\tb^{-S}$.
\begin{cor}\label{lem:normalizationfactor}
Let $\ell$, $S$, $\tb$ such that $\langle 1-P\rangle \leq 1/2$. Then
\[1 \leq N_P \leq 1 + 2e\ell^d (2S+1) \left(C \tb^{-d/2}\right)^{2S} (\log{\ell})^{(3-d)2S}.\]
\end{cor}
\begin{proof}
 We have $N_P = \left(1-\langle 1-P\rangle\right)^{-1} \leq 1+2\langle 1-P\rangle$ as long as $\langle 1-P\rangle \leq 1/2$.
\end{proof}

We are now ready to analyse terms I through IV.

\medskip
\emph{Term I.}
We remove the projection so that we can later calculate the expectation value using Wick's theorem.
\begin{align*}
 \text{I} & = \frac{\langle H^\tD_{0,\lambdal}\rangle_P}{S \lcal^d} = \frac{N_P}{S \lcal^d}  \langle H^\tD_{0,\lambdal}\rangle- \frac{N_P}{S \lcal^d} \langle H^\tD_{0,\lambdal}(1-P) + (1-P)H^\tD_{0,\lambdal}P\rangle.
\end{align*}
Later we are going to show that the kinetic part of $\langle H^\tD_{0,\lambdal}\rangle$
is cancelled by a contribution from term III, and only the expectation value of $I_\lambdal$ remains. Expectation values containing $(1-P)$ are small by the following lemma.
\begin{lem}\label{lem:lem24}
There exists $C > 0$ such that 
\begin{align*}
 &\left\lvert\frac{1}{S\lcal^d} \langle H^\tD_{0,\lambdal}(1-P) + (1-P)H^\tD_{0,\lambdal}P \rangle\right\rvert \leq C \langle 1-P \rangle^{1/2} \sup_{x\in \lambdal} (\rho(x)+1)^{3/2}\rho(x)^{1/2}.
\end{align*}
\end{lem}
\begin{proof}
 By Cauchy-Schwarz
 \begin{equation}\label{eq:oneminuspbound}
 \begin{split}
 &\left\lvert \langle H^\tD_{0,\lambdal}(1-P) + (1-P)H^\tD_{0,\lambdal}P \rangle\right\rvert \leq \langle 1-P \rangle^{1/2} \langle (H^\tD_{0,\lambdal})^2 \rangle^{1/2} +  \langle 1-P \rangle^{1/2} \langle P (H^\tD_{0,\lambdal})^2 P \rangle^{1/2}.
\end{split}
\end{equation}
 Let us denote by $(x,y) \subset \lambdal$ all \emph{ordered} nearest neighbour pairs. Notice that the boundary term can be written as $\sum_{x \in \partial\lambdal} \big(S^2 + S  S^3_x \big) = \sum_{x \in \partial\lambdal} S n_x$. Using Cauchy-Schwarz
 \begin{equation}\label{eq:xx}\begin{split}
 ( H^\text{D}_{0,\lambdal})^2 & \leq S^2 \sum_{(x,y) \subset \lambdal} \sum_{(x',y')\subset\lambdal} \left( -a^*_x a_y + n_x \right) \left( -a^*_{x'} a_{y'} + n_{x'}\right) +  \sum_{x \in \partial\lambdal} \sum_{x' \in \partial\lambdal} S^2 n_x n_{x'}\\
			      & \quad + \sum_{(x,y) \subset \lambdal} \sum_{(x',y')\subset\lambdal} \left(a^*_x\frac{n_x+n_y}{4}a_y - \frac{n_x n_y}{2}\right) \left( a^*_{x'} \frac{n_{x'}+ n_{y'}}{4} a_{y'} - \frac{n_{x'} n_{y'}}{2} \right).
 \end{split}\end{equation}
We now focus on the second summand of \eqref{eq:oneminuspbound} (the first summand is simpler) and estimate it using \eqref{eq:xx}. The contribution of the first term of \eqref{eq:xx} can by Cauchy-Schwarz be estimated as
 \[
 \begin{split}
 & \langle P S^2 \sum_{(x,y)\subset \lambdal} \sum_{(x',y')\subset\lambdal} a^*_x a_y a^*_{x'} a_{y'} P \rangle  \leq 2d S^2 \lcal^d \sum_{(x,y)\subset \lambdal} \langle P  (n_x n_y + n_x) P \rangle;
 \end{split}
\]
the other contributions of \eqref{eq:xx} similarly. Since $P$ commutes with all the operators $n_x$, we can now drop it for an upper bound. Now let us extend the definition of $\rho$ to
\[\rho(x,y) := \langle a^*_y a_x \rangle = \tr a^*_y a_x \frac{e^{-\tb \tdl}}{\tr e^{-\tb \tdl}} \quad \text{(so that $\rho(x)=\rho(x,x)$)}.\]
Then Wick's theorem, followed by Cauchy-Schwarz $\lvert \rho(x,y)\rvert \leq \rho(x)^{1/2} \rho(y)^{1/2}$, yields
 \[
 \begin{split}
 2d S^2 \lcal^d \sum_{(x,y)\subset \lambdal} \langle n_x n_y + n_x \rangle 
 & \leq 2d S^2 \lcal^d \sum_{(x,y)\subset \lambdal} \left( \rho(x)\rho(y)+\rho(x,y)\rho(y,x)+ \rho(x) \right)\\
 & \leq 8d^2 S^2 \lcal^{2d} \sup_{x \in \lambdal} \left( \rho(x)+1 \right)\rho(x). \qedhere
 \end{split}
\]
\end{proof}
\emph{Term II.} This is an error term of order $S^{-2}$. 
As the only error term that is not exponentially small, it constitutes the biggest error in our main theorem.
\begin{lem}\label{lem:II}There exists $C>0$ such that
\[ \lvert \text{II}\rvert = \left\lvert \frac{1}{\lcal^d} \langle R_\lambdal \rangle_P \right\rvert \leq \frac{N_P}{S^2} C  \sup_{x\in\lambdal} (\rho(x)+1) \rho(x)^2.\]
\end{lem}
\begin{proof}
According to \eqref{eq:formalexpansion}, on the subspace with at most $2S$ bosons per site
\[R_\lambdal = \sum_{(x,y) \subset \lambdal} a^*_x \left[\left(1-\frac{n_x}{4S}-\frac{n_y}{4S} \right)- \sqrt{1-\frac{n_x}{2S}}\sqrt{1-\frac{n_y}{2S}} \right] a_y =: \sum_{(x,y)\subset \lambdal} a^*_x A_{x,y} a_y.\]
We have again shortened the expression by writing it as a sum over all \emph{ordered} nearest-neighbour pairs $(x,y)$.

As an operator $A_{x,y} \geq 0$; to see this, notice that it depends only on $n_x$ and $n_y$, which can be diagonalized simultaneously.
 Consequently (for $\psi$ a vector with at most $2S$ particles per site)
\begin{align*}
 \lvert \langle \psi, R_\lambdal \psi \rangle\rvert & =  \bigg\lvert \sum_{(x,y) \subset \lambdal} \langle \psi, a^*_x A_{x,y}^{1/2} A_{x,y}^{1/2} a_y \psi\rangle\bigg\rvert
  \leq \sum_{(x,y) \subset \lambdal} \langle \psi, a^*_x A_{x,y} a_x \psi\rangle,
\end{align*}
and analogously for expectation values $\langle \cdot \rangle_P$.
Again diagonalizing $n_x$ and $n_y$ simultaneously, and using that $\sqrt{1-t} \geq 1 - \frac{t}{2} - \frac{t^2}{2}$ for all $t \in [0,1]$, we find
\[A_{x,y} \leq \frac{1}{8S^2}\left(n_x^2+n_y^2\right).\]
We then write $a^*_x (n_x^2+n_y^2) a_x = n_x(n_x-1)^2+n_x n_y^2$. Now, since $n_x$ and  $n_y$ both commute with $P$, we can drop the $P$s for an upper bound, arriving at
\[
\begin{split}
\frac{1}{\lcal^d} \langle R_\lambdal \rangle_P
 \leq \frac{1}{\lcal^d} \frac{1}{8S^2} \sum_{(x,y)\subset \lambdal} \tr \left( n_x(n_x-1)^2 + n_y^2 n_x\right) \frac{e^{-\tb T^\tD_\lambdal}}{\tr e^{-\tb \tdl}} \frac{\tr e^{-\tb \tdl}}{\tr e^{-\tb \tdl}P}.
\end{split}
\]
Using Wick's theorem this is expanded in terms of $\rho(x)$, $\rho(y)$, and $\rho(x,y)$, and then estimated.
\end{proof}
\emph{Term III.} This splits into a term which cancels the $\tdl$ of Term I, and an error term, i.\,e.\ 
\begin{align*}
 \text{III} & = - \frac{N_P}{\lcal^d}  \langle \tdl \rangle + \frac{N_P}{\lcal^d}  \langle \tdl(1-P) \rangle, \quad \left\lvert \frac{1}{\lcal^d} \langle \tdl(1-P) \rangle \right \rvert \leq  \langle 1-P\rangle^{1/2} \left(\frac{1}{\lcal^{2d}}\langle (\tdl)^2 \rangle \right)^{1/2}.
\end{align*}
Using the momentum space creation/annihiliation operators and Wick's theorem, we find
\[\frac{1}{\ell^{2d}}\langle (\tdl)^2 \rangle = \left(\frac{1}{\ell^d} \sum_{k \in \lambdaldual} \frac{\varepsilon(k)}{e^{\tb \varepsilon(k)}-1} \right)^2 + \frac{1}{\ell^{2d}}\sum_{k\in\lambdaldual} \frac{\varepsilon(k)^2 e^{\tb \varepsilon(k)}}{\left(e^{\tb\varepsilon(k)}-1\right)^2}
\leq C \tb^{-2},\]
by (crudely) estimating $\varepsilon(k)(e^{\tb \varepsilon(k)}-1)^{-1} \leq \tb^{-1} \sup_{x>0} x(e^x-1)^{-1}$, and similarly for the second sum.

\medskip

\emph{Term IV.} We have
\[\text{IV} = -\frac{1}{\tb \lcal^d} \log \tr e^{-\tb \tdl} + \frac{1}{\tb \lcal^d} \log N_P
.\]

\emph{Taking parts I through IV together,} we obtain
\[\begin{split}
   \frac{f^\text{D}(S,\beta,\lambdal)}{S} \leq - \frac{1}{\tb \lcal^d} \log \tr e^{-\tb \tdl} + \frac{1}{\lcal^d} \langle I_{\lambdal} \rangle + \mathcal{E},\\
  \end{split}
\]
 where the terms collected in the error $\mathcal{E}$ can be estimated as
\[\begin{split}
   \lvert\mathcal{E}\rvert & \leq (N_P-1) \frac{\langle I_{\lambdal} \rangle}{\ell^d}+\frac{\log N_P}{\tb \lcal^d}
   + C \langle 1-P\rangle^{1/2} \left[\sup_{x \in \lambdal} \left(\rho(x)+1\right)^{3/2} \rho(x)^{1/2}+N_P \tb^{-1} \right]\\ & \quad + \frac{C N_P}{S^2}\sup_{x\in \lambdal} (\rho(x)+1)^2 \rho(x).
  \end{split}\]
%
Furthermore, the integral approximation of the leading term is
\[\begin{split}
-\frac{1}{\tb \lcal^d} \log \tr e^{-\tb \tdl} & = \frac{1}{\tb \lcal^d} \sum_{k\in \lambdaldual} \log \left( 1-e^{-\tb \varepsilon(k)}\right) \leq \frac{1}{\tb}\int_{[-\pi,\pi]^d} \frac{\di^d k}{(2\pi)^d} \log \left( 1-e^{-\tilde\beta \varepsilon(k)}\right) + \frac{C}{\tb \lcal}  
  \end{split}
\]
and thus, employing Lemmas \ref{lem:rhobound} through \ref{lem:normalizationfactor}, we arrive at the following proposition:
\begin{prp}[Preliminary upper bound]\label{prp:preliminarybound}
The free energy has the upper bound
\[\frac{f^\text{D}(S,\beta,\lambdal)}{S} \leq \frac{1}{\tilde \beta} \int_{[-\pi,\pi]^d} \frac{\di^d k}{(2\pi)^d} \log \left( 1-e^{-\tilde \beta \varepsilon(k)} \right) + \frac{1}{\lcal^d} \langle I_\lambdal \rangle + \mathcal{E}',\]
where the error term $\mathcal{E}'$ satisfies
\[\begin{split}
   \lvert \mathcal{E}' \rvert &\leq \frac{C}{\tilde \beta \ell}+
   C \left[ \ell^d (2S+1) \left(C \tb^{-d/2} \right)^{2S} \log(\ell)^{(3-d)(2S+4)} \right]^{1/2} \left(\langle I_\lambdal \rangle + 1\right) + \frac{C (\log \ell)^{(3-d)3}}{\tb^d S^2}.
  \end{split}
\]
\end{prp}
From the last term we see that the error in the best case will be of order $\tb^{-d}S^{-2}$ (for $d=2$ with a logarithmic correction); to make also the first term and the error from Proposition \ref{prp:twoeight} (see below) that small \emph{we need to choose $\lcal=S^2 \tb^{d}$}. The middle term is exponentially small in $S$ (provided $\tb$ is so large that $C\tb^{-d/2} < 1$, c.\,f.\ Lemma \ref{lem:oneminusp}).

To complete the proof of our theorem, it remains to calculate $\langle I_{\lambdal}\rangle$.

\subsection{Evaluating the Energy Correction}
Here we calculate $\lcal^{-d} \langle I_\lambdal \rangle$. With periodic boundary conditions this is formally simple, but involves infinities for $k=0$. As before, for the rigorous proof we have Dirichlet boundary conditions, making the evaluation somewhat more complicated. In particular, we will find finite-size errors (smaller by $1/\ell$ compared to the leading `bulk' term).
\begin{prp}[First order correction]\label{prp:twoeight} We have
\[\frac{1}{\ell^{d}} \langle I_\lambdal \rangle \leq - \frac{1}{S} \frac{1}{4 d} \left( \int_{[-\pi,\pi]^d} \frac{\di^d k}{(2\pi)^d} \frac{\varepsilon(k)}{e^{\tilde \beta \varepsilon(k)}-1}\right)^2 + \frac{C(\log \ell)^{3-d}}{S\ell}.\]
\end{prp}

\begin{proof}
 Recall that according to \eqref{eq:formalexpansion} we have
 \begin{equation}\label{eq:termI}
  I_\lambdal = \frac{1}{4S}\sum_{\langle x,y\rangle \subset \lambdal} \left( a^*_x a^*_x a_x a_y + a^*_x a^*_y a_y a_y + a^*_y a^*_x a_x a_x + a^*_y a^*_y a_y a_x - 4 a^*_x a^*_y a_x a_y \right).
 \end{equation}
 By Wick's theorem
\[\begin{split} \langle I_\lambdal\rangle  = \frac{1}{S}\sum_{\langle x,y\rangle \subset \lambdal} \left( \rho(x,x)\rho(x,y) + \rho(x,x)\rho(y,x) - \rho(x,x)\rho(y,y) - \rho(x,y) \rho(y,x)\right).
  \end{split}
\]
We use the Fourier representation \eqref{eq:deffouriertransform} of $a^*_x$ and $a_x$ to calculate $\rho(y,x)$, giving
\[\rho(y,x) = \tr a^*_x a_y \gibbs = \sum_{k,k' \in \lambdaldual} \varphi_k(x) \varphi_{k'}(y) \tr a^*_k a_{k'} \gibbs = \sum_{k \in \lambdaldual} \varphi_k(x) \varphi_k(y) \frac{1}{e^{\tilde \beta \varepsilon(k)}-1}.\]
Thus we obtain
(using the abbreviation $f(k)=(e^{\tilde \beta \varepsilon(k)}-1)^{-1}$)
\[\begin{split}
\langle I_\lambdal\rangle &= - \frac{1}{S}\sum_{\langle x,y\rangle \subset \lambdal}\sum_{k,k' \in \lambdaldual} f(k) f(k')\big[\varphi_k(x)^2 \varphi_{k'}(y)^2 - 2 \varphi_k(x)^2 \varphi_{k'}(x) \varphi_{k'}(y)  + \varphi_k(x)\varphi_k(y)\varphi_{k'}(x)\varphi_{k'}(y)\big]\\
&= - \frac{1}{S}\sum_{\substack{x_j=1,\ldots,\ell\\j=1,\ldots,d}} \sum_{i=1}^d \sum_{k,k'\in\lambdaldual} f(k) f(k')\left(\frac{2}{\ell+1}\right)^{2d}\\
&\quad\times \big[\prod_{j=1}^d \sin(x_j k_j)^2 \sin((x_j+\delta_{i,j})k'_j)^2+2\prod_{j=1}^d \sin(x_j k_j)^2 \sin(x_j k'_j) \sin((x_j+\delta_{i,j})k'_j)\\
&\quad\quad+\prod_{j=1}^d \sin(x_j k_j) \sin((x_j+\delta_{i,j})k_j)\sin(x_j k'_j)\sin((x_j+\delta_{i,j})k'_j)\big].
\end{split}\]
A little regrouping of the last expression yields
\begin{equation}\label{eq:prevstep}
\begin{split}
   \langle I_\lambdal\rangle & = - \frac{1}{S}\sum_{k,k'\in \lambdaldual} f(k)f(k') \left(\frac{2}{\ell+1}\right)^d \sum_{i=1}^d \prod_{j\in\{1,\ldots,d\}\backslash\{i\}} \left( \frac{2}{\ell+1} \sum_{x_j=1}^\ell\sin(x_j k_j)^2\sin(x_j k'_j)^2\right)\\
   &\quad \times \frac{2}{\ell+1}\sum_{x_i=1}^\ell\big[ \sin(x_i k_i)^2 \sin((x_i+1)k'_i)^2 - 2 \sin(x_i k_i)^2 \sin(x_i k'_i)\sin((x_i+1)k'_i)\\
   &\hspace{3cm}+\sin(x_i k_i)\sin((x_i+1)k_i)\sin(x_i k'_i)\sin((x_i+1)k'_i) \big].
  \end{split}
\end{equation}
%
%
Let us for the moment look only at the last factor of \eqref{eq:prevstep}, and expand it using $\sin(x_i k_i + k_i)=\sin(x_i k_i)\cos(k_i)+\cos(x_i k_i)\sin(k_i)$. Using Lemma \ref{lem:vanishing}(i) we find
\begin{align*}
& \frac{2}{\ell+1}\sum_{x_i=1}^{\ell}\big[ \sin(x_i k_i)^2 \sin((x_i+1)k'_i)^2 - 2 \sin(x_i k_i)^2 \sin(x_i k'_i)\sin((x_i+1)k'_i)\nn\\
   &\hspace{2.2cm}+\sin(x_i k_i)\sin((x_i+1)k_i)\sin(x_i k'_i)\sin((x_i+1)k'_i) \big]\nn\\
   & = \frac{2}{l+1}\sum_{x_i=1}^{\ell} \big[\sin(x_i k_i)^2 \sin(x_i k'_i)^2 \cos(p'_i)^2
   + \sin(x_i k_i)^2 \cos(x_i k'_i)^2 \sin(k'_i)^2\\
   & \hspace{2.7cm}- 2 \sin(x_i k_i)^2 \sin(x_i k'_i)^2 \cos(k'_i)
   + \sin(x_i k_i)^2 \sin(x_i k'_i)^2 \cos(k_i) \cos(k'_i)\\
   & \hspace{2.7cm}+ \sin(x_i k_i) \cos(x_i k_i) \sin(x_i k'_i)\cos(x_i k'_i)\sin(k_i)\sin(k'_i)\big].
\end{align*}
Thereafter we further evaluate \eqref{eq:prevstep} using Lemma \ref{lem:quartictrigonometry}(ii). Then we use the symmetry between $k$ and $k'$ to write $-2\cos(k'_i)$ as $-\cos(k_i)-\cos(k'_i)$. After these steps we have
\[
\begin{split}
 \langle I_{\lambdal} \rangle & = - \frac{1}{S}\sum_{k,k' \in \lambdaldual} f(k) f(k') \left(\frac{2}{\ell+1}\right)^d\sum_{i=1}^d\prod_{j\in\{1,\ldots,d\}\backslash \{i\}} \left( \frac{1}{2} + \frac{1}{4}\left(\delta_{k_j,k'_j} + \delta_{k_j+k'_j,\pi} \right) \right)\\
  &\quad\quad\quad \times \frac{1}{2}\bigg\{ \left(1-\cos(k_i)\right)\left(1-\cos(k'_i)\right)
  +\delta_{k_i,k'_i}\left(\cos(k_i)^2-\cos(k_i)\right) + \delta_{k_i+k'_i,\pi}\cos(k_i)
  \bigg\}.
\end{split}
\]
Next we use symmetry among $k_1, \ldots, k_d$ to replace the sum over $i=1,\ldots,d$ by a factor of $d$. Since deltas eliminate a sum, every factor of a delta\footnote{These deltas constitute the finite-size errors mentioned before, and are a consequence of using Dirichlet boundary conditions. (They do not appear when formally using periodic boundary conditions.)} is effectively of order $1/\ell$. However, for $d=2$, some of the sums appearing are $\log \ell$-divergent at small $k$ and $k'$.
The resulting estimate is
\begin{equation}\label{eq:result}
\begin{split}
\frac{1}{\ell^{d}}\langle I_{\lambdal} \rangle & \leq - \frac{d}{S\ell^{d}} \sum_{k,k' \in \lambdaldual} f(k) f(k') \left(\frac{2}{\ell+1}\right)^d \frac{1}{2^d}\left(1-\cos(k_1)\right)\left(1-\cos(k'_1)\right) + C S^{-1}(\log \ell)^{3-d}\ell^{-1}\\
&= -\frac{d}{S} \Bigg[\frac{1}{(\ell+1)^d}\sum_{k \in \lambdaldual} f(k) \left(1-\cos(k_1) \right) \Bigg]^2  + C S^{-1}(\log \ell)^{3-d}\ell^{-1}.
\end{split}
\end{equation}
Again we use the symmetry among $k_1,\ldots,k_d$, now to replace $(1-\cos(k_1))$ by $\frac{1}{2d}\varepsilon(k)$.

It remains to employ the continuum approximation for $g(k) := f(k)\varepsilon(k)$: Lemma \ref{lem:continuumapproximation} with $D_1 = \tb^{-1} \sup_{x\in\Rbb} x(e^x-1)^{-1}$ and $D_2 = \sup_{k\in \lambdaldual} \lvert \nabla g(k)\rvert \leq 2\sqrt{3} \sup_{x\in\Rbb} \lvert x e^x +1 - e^x\rvert(e^x-1)^{-2}$ implies 
\[\frac{1}{\ell^{d}} \langle I_{\lambdal}\rangle \leq -\frac{1}{4 d S}\left[ \int_{[0,\pi]^d} \frac{\di^d k}{\pi^d} \frac{\varepsilon(k)}{e^{\tb \varepsilon(k)}-1} \right]^2 + CS^{-1} (\log \ell)^{3-d}\ell^{-1}.\]
Since $\varepsilon(k)$ depends on $k$ only through $\cos(k_i)$, the integral remains unchanged by any reflection $k_i \mapsto -k_i$.  Thus the integral over $[0,\pi]^d$ is the same as $2^{-d}$ times the integral over $[-\pi,\pi]^d$.
\end{proof}

The next two lemmas were used in the previous proof.
\begin{lem}\label{lem:vanishing}\label{lem:quartictrigonometry}
Let $k, k' \in \lambdaldual$ and $i=1,\ldots d$.
\begin{enumerate}
 \item  We have
\[\sum_{x_i=1}^{\ell} \sin(x_i k_i)^2 \sin(x_i k'_i) \cos(x_i k'_i) = 0.\]
\item Furthermore
 \begin{align*}& \frac{2}{\ell+1}\sum_{x_i=1}^{\ell} \sin(x_i k_i)^2 \sin(x_i k'_i)^2 = \frac{1}{2}+\frac{1}{4}\left( \delta_{k_i,k'_i} + \delta_{k_i+k'_i,\pi}\right),\\
 & \frac{2}{\ell+1}\sum_{x_i=1}^{\ell} \sin(x_i k_i)^2 \cos(x_i k'_i)^2 = \frac{1}{2}-\frac{1}{4}\left( \delta_{k_i,k'_i} + \delta_{k_i+k'_i,\pi} \right),\\
 & \frac{2}{\ell+1}\sum_{x_i=1}^{\ell} \sin(x_i k_i) \cos(x_i k_i) \sin(x_i k'_i) \cos(x_i k'_i) = \frac{1}{4}\left( \delta_{k_i,k'_i} - \delta_{k_i+k'_i,\pi} \right).
 \end{align*}
\end{enumerate}
\end{lem}
\begin{proof}
(i) Expanding the trigonometric functions in terms of exponentials we find
\begin{align*}
&\sum_{x_i=1}^\ell \sin(x_i k_i)^2 \sin(x_i k'_i) \cos(x_i k'_i) = \sum_{x_i=1}^\ell \left[\frac{\sin(2x_i(k_i-k'_i))-\sin(2 x_i(k_i+k'_i))}{4} + \frac{\sin(2 x_i k'_i)}{2} \right].
\end{align*}
Due to the factor of $2$ in all the arguments on the r.\,h.\,s., the summation range  contains only multiples of full phases of the sine (recall that $k_i = \frac{\pi}{\ell+1}n$ for some $n \in \Zbb$), so negative and positive parts cancel. (ii) We simply expand into exponentials and use the finite geometric sum.
\end{proof}

\begin{lem}\label{lem:continuumapproximation}
 Let $g: (0,\pi]^n \to \Rbb$ be bounded above by some $D_1 < \infty$
 and have Lipschitz constant $D_2 < \infty$.
 Then there exists a constant $C > 0$ (depending only on the dimension $n$) such that
 \[\left(\frac{\pi}{\ell+1}\right)^n \sum_{k \in \lambdaldual} g(k) \geq \int_{[0,\pi]^n} g(k) \di^n k - C (D_1+D_2) \frac{1}{\ell+1} \quad \forall \ell \in \Nbb.\]
\end{lem}
\begin{proof} Due to Lipschitz continuity, for every $k_0 \in \lambdaldual$ we have
 \[\left( \frac{\pi}{\ell+1} \right)^n \left[ g(k_0) + D_2\sqrt{n}\frac{\pi}{\ell+1} \right] \geq \max_{k \in k_0+[0,\pi/(\ell+1)]^n} g(k) \left( \frac{\pi}{\ell+1} \right)^n \geq \int_{k_0+[0,\pi/(\ell+1)]^n} g(k)\di^n k.\]
 Thus (the factor $\ell^n$ in the numerator of the last summand being the number of boxes in the partition, or equivalently, the number of elements of $\lambdaldual$)
 \[\left(\frac{\pi}{\ell+1} \right)^n \sum_{k_0 \in \lambdaldual} g(k_0) \geq \int_{[\pi/(\ell+1),\pi]^n} g(k)\di^n k -  \frac{\pi^n\ell^n}{(\ell+1)^n} D_2 \sqrt{n} \frac{\pi}{\ell+1}.\]
 Since $g$ is bounded above, by extending the integration range from $[\pi/(\ell+1),\pi]^n$ to $[0,\pi]^n$, we make the integral larger by a quantity of at most $C D_1/(\ell+1)$ for some $C < \infty$ depending only on $n$.
\end{proof}

\appendix

\section{Appendix: Cancellation at Second Order}
\label{sec:appendix} We now consider the three dimensional case only. 
As explained in the beginning, the dependence of our error bound on $\tb$ (order $\tb^{-3}$) is not in agreement with Dyson's paper, which claims that all corrections are $\tb^{-5}$ and smaller. However, it was pointed out \cite{Oguchi} that at second order of formal perturbation theory there is a cancellation, by which Dyson's result is reproduced to order $1/S^2$ (if one corrects for a trivial numerical imprecision \cite{PD66}).

\medskip

Below we reproduce the calculation of \cite{Oguchi} in the language of modern perturbation theory and in detail. For simplicity we work in periodic boundary conditions, i.\,e.\ the eigenfunctions of the Laplacian are $\ell^{-3/2} e^{i k\cdot x}$ with momenta\footnote{All $k$-sums in this appendix are over this range.} $k \in \frac{2\pi}{\ell}\left(\Zbb^{3} \cap (-\ell,\ell)^3\right)$. We expand the Hamiltonian one order further,
\[H_{\Lambda_\ell} =: S\left( T_{\Lambda_\ell} + I_{\Lambda_\ell} + J_{\lambdal} + \tilde{R}_{\Lambda_\ell} \right), \quad \text{i.\,e.}\quad R_\lambdal = J_{\lambdal} + \tilde{R}_{\Lambda_\ell}.\]
Here $J_\lambdal$ is formally of order $S^{-2}$ and after normal-ordering given by
\[J_{\lambdal} = \frac{1}{32S^2} \sum_{(x,y) \subset \lambdal} \left( a^*_x a^*_y a^*_y a_y a_y a_y + a^*_x a^*_y a_y a_y - 2 a^*_x a^*_x a^*_y a_x a_y a_y + a^*_x a^*_x a^*_x a_x a_x a_y + a^*_x a^*_x a_x a_y \right).\]
Using Wick's theorem we find the correction to $f(S,\beta,\lambdal)/S$ to be
\begin{equation}\label{eq:firstorderpert}\frac{\langle J_\lambdal \rangle}{\ell^3} = \frac{1}{4 S^2}\sum_{i=1}^3 \bigg[ (\rho+\rho^2) \frac{1}{\ell^3} \sum_{k} f(k) \cos(k_i) - \frac{1}{\ell^9} \sum_{k_1,k_2,k_3} f(k_1) f(k_2) f(k_3)\cos(k_{1,i}-k_{2,i}+k_{3,i})\bigg],\end{equation}
where $\rho = \frac{1}{\ell^3}\sum_{k} f(k)$ and $f(k) = 1/(e^{\tb \varepsilon(k)}-1)$. The biggest contribution\footnote{Strictly speaking this sum is infinite because the contributions of $k_1=0$ and $k_2=0$ are infinite as a remnant of using periodic boundary conditions. The cancellation below also resolves this issue.} is
\begin{equation}\label{eq:biggesterror}\frac{\rho}{4 S^2} \frac{1}{\ell^3} \sum_{k} f(k) \sum_{i=1}^3 \cos(k_i) = \frac{1}{16 S^2 \ell^9} \sum_{k_1,k_2,k_3} f(k_1) f(k_2) \left[12-\varepsilon(k_1)-\varepsilon(k_2) \right] \sim \tb^{-3}.\end{equation}
The remaining part of \eqref{eq:firstorderpert} is finite and of order $\tb^{-11/2}$, as can be seen by expanding the cosines, observing that the lowest terms cancel, replacing the sum by an integral for $\ell \to \infty$ and using the scaling $\int \di^3 k\, ( \exp(\tb \lvert k\rvert^2) - 1 )^{-1} \sim \tb^{-3/2}$. 

The big error \eqref{eq:biggesterror} originates from the summands in $J_{\lambdal}$ with four creation and annihilation operators (which all originate from normal-ordering of the formal expansion). Thus as an operator bound on $R_\lambdal$ we can not expect an estimate better than $\tb^{-3}$. Instead we have to take into account the structure of the interacting Gibbs state. We verify this by showing that in second order of perturbation theory, \eqref{eq:biggesterror} is cancelled up to a $\tb^{-5}$-remainder.

\medskip

\emph{Second order perturbation theory.} The second order perturbation theory is given through the Duhamel formula as
\begin{align}
\frac{f(S,\beta,\Lambda_\ell)}{S}  = - \frac{\log \tr e^{-\beta H_\lambdal}}{\tb \ell^3} & = - \frac{\log \tr e^{-\tb T_\lambdal}}{\tb \ell^3} + \frac{1}{\ell^3} \langle I_\lambdal+J_\lambdal \rangle\nn\\
& \quad - \frac{1}{\tb \ell^3}\int_0^{\tb}\di s \int_0^{s} \di s' \left( \langle I_\lambdal(s) I_\lambdal(s')\rangle -  \langle I_\lambdal\rangle \langle I_\lambdal\rangle \right) + \order(S^{-3}),\label{eq:secondorder} \end{align}
where $I_\lambdal(s) := e^{sT_\lambdal}I_\lambdal e^{-sT_\lambdal}$.
The contractions in momentum space (represented in the Feynman diagrams by lines with arrow pointing from the creation to the annihilation operator) are
\begin{equation}\label{eq:contractions}\begin{split}\langle a^*_k(s)a_{k'}(s')\rangle = \delta_{k,k'} e^{(s-s')\varepsilon(k)}f(k)\quad \text{and}\quad\langle a_k(s)a^*_{k'}(s')\rangle = \delta_{k,k'} e^{-(s-s')\varepsilon(k)}(1+f(k)).\end{split}\end{equation}
Using $a^*_x = \ell^{-3/2} \sum_k e^{ik\cdot x} a^*_k$, we find the interaction operator $I_{\lambdal}$ in momentum space to be
 \[\begin{split}I_\lambdal = 
\frac{1}{8 S \ell^3} \!\sum_{k_1,k_2,k_3,k_4}\! &  \delta_{k_1+k_2,k_3+k_4} \nu(k_1,k_2,k_3,k_4)a^*_{k_1} a^*_{k_2} a_{k_3} a_{k_4},\end{split}\]
where\footnote{We abbreviate $\varepsilon(k_1) = \varepsilon_1$, $\varepsilon(k_1-k_2) = \varepsilon_{1-2}$ etc.} $\nu(k_1,k_2,k_3,k_4) := \varepsilon_{4-2}-\varepsilon_4 - \varepsilon_1 + \varepsilon_{4-1} - \varepsilon_2 + \varepsilon_{3-2} + \varepsilon_{3-1} -\varepsilon_3$. Notice that we have symmetrized $\nu$ w.\,r.\,t.\ exchange of $k_1$ with $k_2$ or $k_3$ with $k_4$, so in Feynman diagrams $I_{\lambdal}$ is represented as a vertex with two equivalent outgoing and two equivalent ingoing legs.

The second order contribution \eqref{eq:secondorder} can through Wick's theorem be represented by the sum of all connected Feynman diagrams having two $I_{\lambdal}$-vertices, where the numerical factors count the number of equivalent diagrams:
\vspace{-.5em}
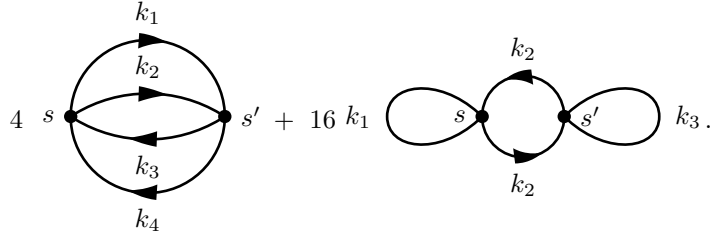
\begin{equation*}\label{eq:diagrams}4\begin{gathered}
\begin{fmfgraph*}(90,90)
\fmfleft{i}
\fmfright{o}
\fmf{phantom,tension=10}{i,v1}
\fmf{phantom,tension=10}{v2,o}
\fmflabel{$s$}{v1}
\fmflabel{$s'$}{v2}
\fmf{fermion,left,tension=0.4,label=$k_1$}{v1,v2}
\fmf{fermion,left,tension=0.4,label=$k_4$}{v2,v1}
\fmf{fermion,left=0.3,label=$k_2$}{v1,v2}
\fmf{fermion,left=0.3,label=$k_3$}{v2,v1}
\fmfdot{v1}
\fmfdot{v2}
\end{fmfgraph*}
\end{gathered}
+\ 16\ \hspace{2em} \begin{gathered}
\begin{fmfgraph*}(90,90)
\fmftop{t0,t1,t2,t3}
\fmfbottom{b0,b1,b2,b3}
\fmf{phantom}{t1,v1,b1}
\fmf{phantom}{t2,v2,b2}
\fmffreeze
\fmflabel{$s$}{v1}
\fmflabel{$s'$}{v2}
\fmf{fermion,right,label=$k_2$}{v1,v2}
\fmf{fermion,right,label=$k_2$}{v2,v1}
\fmf{plain,tension=0.7,right=270,label=$k_1$}{v1,v1}
\fmf{plain,tension=0.7,right,label=$k_3$}{v2,v2}
\fmfdot{v1,v2}
 \end{fmfgraph*}
\end{gathered}\hspace{2.3em}.\vspace{-.8em}\end{equation*}
\emph{Left Feynman diagram.} Using the contractions \eqref{eq:contractions} we write out the left Feynman diagram
\[\begin{split}
& -\frac{4}{\tb\ell^3} \int_0^\tb\di s\int_0^s \di s'\frac{1}{\left(8S\ell^3\right)^2}
\sum_{\substack{k_1,k_2,k_3,k_4\\k_1',k_2',k_3',k_4'}} \delta_{k_1+k_2,k_3+k_4} \delta_{k'_1+k'_2,k'_3+k'_4}
\nu(k_1,k_2,k_3,k_4)\nu(k'_1,k'_2,k'_3,k'_4)\\
&\quad\times\delta_{k_1,k_4'}\delta_{k_2,k_3'}\delta_{k_3,k_2'}\delta_{k_4,k_1'}e^{(s-s')\varepsilon_1}e^{(s-s')\varepsilon_2}e^{-(s-s')\varepsilon_3} e^{-(s-s')\varepsilon_4} f(k_1) f(k_2) (1+f(k_3)) (1+f(k_4)).
\end{split}\]
The integrations over $s'$ and $s$ are done explicitly, yielding
\[\begin{split}
& -\frac{1}{16\tb S^2 \ell^9} \sum_{k_1,k_2,k_3,k_4}\delta_{k_1+k_2,k_3+k_4}\nu(k_1,k_2,k_3,k_4)^2\\
&\quad\times\left[-\frac{\tb}{\varepsilon_1+\varepsilon_2-\varepsilon_3-\varepsilon_4}+ \frac{e^{\tb(\varepsilon_1+\varepsilon_2-\varepsilon_3-\varepsilon_4)}-1}{(\varepsilon_1+\varepsilon_2-\varepsilon_3-\varepsilon_4)^2} \right] \frac{e^{-\tb\varepsilon_1}}{1-e^{-\tb\varepsilon_1}} \frac{e^{-\tb\varepsilon_2}}{1-e^{-\tb\varepsilon_2}} \frac{1}{1-e^{-\tb\varepsilon_3}} \frac{1}{1-e^{-\tb\varepsilon_4}}.\end{split}\]
We observe that
$[e^{\tb(\varepsilon_1+\varepsilon_2-\varepsilon_3-\varepsilon_4)}-1]e^{-\tb\varepsilon_1}e^{-\tb\varepsilon_2}$ has odd sign under simultaneous exchange of $k_1$ with $k_3$ and $k_2$ with $k_4$, and is multiplied with an expression of even sign, thus summing to zero. We are left with the contribution of the first summand from the square brackets,
\begin{equation}\label{eq:leftdiagram}\begin{split}
& \frac{1}{16 S^2 \ell^9} \sum_{k_1,k_2,k_3,k_4}\delta_{k_1+k_2,k_3+k_4}\frac{\nu(k_1,k_2,k_3,k_4)^2}{\varepsilon_1+\varepsilon_2-\varepsilon_3-\varepsilon_4}f(k_1)f(k_2)(1+f(k_3))(1+f(k_4))\\
& = \frac{1}{16 S^2 \ell^9}\sum_{k_1,k_2,k_3} \frac{\left(2\varepsilon_{1-3}+2\varepsilon_{2-3}-\varepsilon_1-\varepsilon_2-\varepsilon_3-\varepsilon_{1+2-3} \right)^2}{\varepsilon_1+\varepsilon_2-\varepsilon_3-\varepsilon_{1+2-3}} \left[f(k_1)f(k_2) + 2 f(k_1) f(k_2) f(k_3) \right].\end{split}\end{equation}
(In the last line the contribution $f(k_1)f(k_2)f(k_3)f(k_4)$ disappeared by an antisymmetry argument.)

\emph{Right Feynman diagram.} The right Feynman diagram is evaluated similarly, and found to be
\begin{equation}\label{eq:rightdiagram}\begin{split}- \frac{\tb}{2 S^2 \ell^9}\sum_{k_1,k_2,k_3} \left(\varepsilon_{2-3} - \varepsilon_3-\varepsilon_2 \right) \left(\varepsilon_{2-1} - \varepsilon_1-\varepsilon_2 \right)
f(k_1) f(k_2) \left(1+f(k_2) \right) f(k_3).
\end{split}\end{equation} 

 Replacing the sum by an integral as $\ell \to \infty$ and expanding $\varepsilon(k) \simeq \lvert k\rvert^2$, we find that \eqref{eq:rightdiagram} is finite and of order $\tb^{-11/2}$ (so for us negligible). By the same argument, the part of \eqref{eq:leftdiagram} containing $f(k_1)f(k_2)f(k_3)$ is of order $\tb^{-11/2}$ (and thus also negligible).
 
 \medskip
 
 \emph{The cancellation.} As possibly big contributions remain the part of \eqref{eq:leftdiagram} containing only $f(k_1)f(k_2)$ and \eqref{eq:biggesterror}. The crucial observation \cite{Oguchi} is that these two terms cancel up to terms of order $\tb^{-5}$! Let us spell out the argument more carefully.
The sum of the two terms is
 \[\begin{split}& \frac{1}{16 S^2 \ell^9} \sum_{k_1,k_2,k_3} f(k_1)f(k_2) \left[ 12-\varepsilon_1-\varepsilon_2 + \frac{\left[ (\varepsilon_1+\varepsilon_2-\varepsilon_3-\varepsilon_{1+2-3})+2(\varepsilon_3 +  \varepsilon_{1+2-3} - \varepsilon_{2-3}-\varepsilon_{1-3}) \right]^2}{\varepsilon_1+\varepsilon_2 - \varepsilon_3 - \varepsilon_{1+2-3}} \right]\\
 & = \frac{1}{16 S^2 \ell^9} \sum_{k_1,k_2,k_3} f(k_1)f(k_2) \left[ 12 + 3\varepsilon_3 + 3\varepsilon_{1+2-3} - 4\varepsilon_{2-3}- 4\varepsilon_{1-3} + 4 \frac{(\varepsilon_3 +  \varepsilon_{1+2-3} - \varepsilon_{2-3}-\varepsilon_{1-3})^2}{\varepsilon_1+\varepsilon_2 - \varepsilon_3 - \varepsilon_{1+2-3}}\right]
   \end{split}
\]
 Recalling that $\varepsilon(k) = \sum_{j=1}^3 2(1-\cos(k_j))$, the first part vanishes since the sum
 is over all $k$, i.\,e.\ positive and negative parts of the cosine sum to zero:
 \[\begin{split}& \sum_{k_3} \left[ 12 + 3\varepsilon_3 + 3\varepsilon_{1+2-3} - 4\varepsilon_{2-3}- 4\varepsilon_{1-3} \right] =0.
 \end{split}\] 
 Concerning the second part, by expanding $\varepsilon(k)$, we find
 \[\begin{split}& \frac{1}{16 S^2 \ell^9} \sum_{k_1,k_2,k_3} f(k_1)f(k_2)\,4 \frac{(\varepsilon_3 +  \varepsilon_{1+2-3} - \varepsilon_{2-3}-\varepsilon_{1-3}))^2}{\varepsilon_1+\varepsilon_2 - \varepsilon_3 - \varepsilon_{1+2-3}}\\&\simeq
 \frac{1}{4S^2 \ell^9} \sum_{k_1,k_2,k_3} f(k_1)f(k_2)\frac{(k_1\cdot k_2)^2}{(k_1-k_3)\cdot (k_2-k_3)} \sim \tb^{-5} S^{-2}.\end{split}\]
 This agrees with Dyson's result at order $1/S^2$. 

Presumably such cancellations appear at all orders in perturbation theory. In Dyson's work the problematic quartic  terms in $J_{\lambdal}$ from normal-ordering do not appear, at the cost of working with a non-selfadjoint Hamiltonian. While this approach is supposed to be equivalent \cite{OtherOguchi} to the Holstein-Primakoff approach followed here, it  has never been made completely rigorous. It remains an interesting problem to rigorously obtain an upper bound with optimal $\tb$-dependence.

\section*{Acknowledgements}
I would like to thank Alessandro Giuliani and Michele Correggi for many discussions and the introduction to the topic. I acknowledge support by ERC grant CoMBoS-239694, ERC Advanced grant 321029, by MIUR through the FIR grant 2013 ``Condensed Matter in Mathematical Physics (Cond-Math)'' (code RBFR13WAET), and by VILLUM FONDEN via the QMATH Centre of Excellence (Grant No. 10059).

\bibliographystyle{plain}
\bibliography{heisenberg}
\end{fmffile}
\end{document}